\documentclass[12pt]{article} 

\usepackage[utf8]{inputenc} 

\usepackage{hyperref}
\usepackage{cite}
\hypersetup{colorlinks=false}

\usepackage{geometry,amsmath} 
\usepackage{graphicx} 

\numberwithin{equation}{section}

\usepackage{booktabs} 
\usepackage{array} 
\usepackage{paralist} 
\usepackage{verbatim} 
\usepackage{subfig} 

\usepackage{float}

\usepackage{todonotes}

\usepackage{amssymb}
\usepackage{amsmath}
\usepackage{amsthm}
\usepackage{mathrsfs}

\usepackage{xcolor,pict2e}

\theoremstyle{plain}
\newtheorem{thm}{Theorem}[section]

\newtheorem{exm}[thm]{Example}
\newtheorem{defn}[thm]{Definition}
\newtheorem{rem}[thm]{Remark}

\usepackage{fancyhdr} 
\usepackage{sectsty}
\allsectionsfont{\sffamily\mdseries\upshape} 

\title{A modified formal Lagrangian formulation for general differential equations  }

\author{Linyu Peng\footnote{Adjunct faculty member at the School of Mathematics and Statistics, Beijing Institute of Technology, China, and adjunct researcher at the Waseda Institute for Advanced Study, Waseda University, Japan. Email: l.peng@mech.keio.ac.jp } \vspace{0.4cm}
\\
{\it Department of Mechanical Engineering, Keio University,}
 \\
{\it Yokohama 223-8522, Japan}}

\begin{document}

\maketitle
\abstract{In this paper, we propose a modified formal Lagrangian formulation by introducing dummy  dependent variables and prove the existence of such a formulation for any system of  differential equations.  The corresponding Euler--Lagrange equations, consisting of the original system and its adjoint system about the dummy variables, reduce to the original system via a simple substitution for the dummy variables. The formulation is applied to study conservation laws of differential equations through Noether's Theorem and in particular, a nontrivial conservation law of the Fornberg--Whitham equation is obtained by using its Lie point symmetries. Finally, a correspondence between conservation laws of the incompressible Euler equations and variational symmetries of the relevant modified formal Lagrangian is shown. 

}




%
%

\section{Introduction}
A lot of differential equations arising from physical phenomena can be derived from variational calculus, that studies extrema of functionals, i.e., functions whose arguments are also functions. Variational structure not only allows us to study  geometric properties of differential equations systematically (e.g., \cite{And1989,KraVin1999,MR1999a}), but also serves as an important framework for the development of geometric numerical  integrators (e.g., \cite{HLW2006,MPS1998,MW2001}). Another great advantage of a variational structure is that Noether's Theorem can immediately be applied to derive conservation laws of variational differential equations. Noether's Theorem, establishing a one-to-one correspondence between variational symmetries and conservation laws of the Euler--Lagrange equations, was proved by Emmy Noether and published in 1918 \cite{Noe1918};  see \cite{Olver1993} for a modern version and \cite{KS2011} for a history of Noether's Theorem  together with her second theorem. 

Conservation laws are among the most important properties of differential equations. To apply Noether's Theorem, it is necessary to study inverse problems, namely, to distinguish variational differential equations from the others, and to find the corresponding functional when the system is variational. Unfortunately, a variational structure is not always available for general differential equations. Many methods for deriving  conservation laws of general or special type of  differential equations, nevertheless,  have been developed, for instance, Vinogradov's $\mathscr{C}$-spectral sequence  \cite{KraVin1999,Vin1984a,Vin1984b},  symbolic methods \cite{GH1997,Her2006}, the direct construction method of Anco \& Bluman \cite{AB2002a,AB2002b}, using partial Lagrangians \cite{KM2006}, the formal variational structure and self-adjointness method \cite{AH1975,Gan2011,Ibr2007,Ibr2011a}, etc. 

In particular, the formal variational structure was firstly proposed by Atherton \& Homsy \cite{AH1975}  in 1975 where they called the composite principles of differential equations; see also \cite{Olver1993}. Atherton \& Homsy showed that a composite principle can always be formulated for any general system of differential equations. Their formulation was then developed by Ibragimov \cite{Ibr2007}, by studying the self-adjointness of the Euler--Lagrange equations consisting of the original system and its adjoint system.   A formal Lagrangian $L$ can be defined for a general system of differential equations $\{F_{\alpha}=0,\quad \alpha=1,2,\ldots,l\}$  by introducing dummy dependent variables $v$, namely
\begin{equation*}
L:=\sum_{\alpha}v^{\alpha}F_{\alpha}.
\end{equation*}
Symmetries of the original system can be extended to variational symmetries of the corresponding formal variational problem, and hence conservation laws of the Euler--Lagrange equations, which include both the original system and its adjoint system, can be derived using Noether's Theorem. If, through a proper  substitution for the dummy variables, the Euler--Lagrange equations reduce to the original system, namely, the adjoint system is equivalent to the original system, called the self-adjointness of the system, then the so-obtained 
  conservation laws turn into conservation laws of the original system through the same substitution. Although these conservation laws can sometimes be trivial, this method  provides a straightforward algorithm for computing conservation laws of non-variational differential equations using Noether's Theorem. Furthermore, it also makes the development of variational integrator for non-variational differential systems possible; see, e.g., \cite{KM2015}. Some studies on extensions of such techniques to discrete equations are available, e.g., \cite{Peng2015,Peng2017}.

Even though every system can be embedded into a bigger variational system, without the self-adjointneess it is difficult to obtain conservation laws of the original system by simply using Noether's Theorem.  It was realised, unfortunately, that self-adjointness of many differential equations can not be expressed using  simple/unified substitutions for the dummy variables. More complex substitutions were introduced and successfully applied to some differential equations,  for instance, the weak self-adjointness \cite{Gan2011} and nonlinear self-adjointness \cite{Ibr2011a}. However, they are often case-by-case depending on the system of interest. In this paper, we propose a systematical method that is applicable to study conservation laws of all differential equations by deriving their self-adjointness through the simplest substitution for dummy variables $v$. This is made possible  by modifying a formal Lagrangian by adding an extra so-called balance function $L_0$, which is independent from the dummy variables:
\begin{equation*}
\widehat{L}:=\sum_{\alpha}v^{\alpha}F_{\alpha}+L_0
\end{equation*}
This  method will be called  a {\it modified formal Lagrangian formulation}.  First of all, it should be noted that the balance function $L_0$ can always be systematically constructed. Secondly, the corresponding modified Euler--Lagrange equations reduce to the original system $\{F_{\alpha}=0,\quad \alpha=1,2,\ldots,l\}$ by the simplest substitution for the dummy variables via $v=u$. By extending symmetries of the original system---which keep the balance function $L_0$ invariant---to variational symmetries of the modified formal Lagrangian $\widehat{L}$, conservation laws of the modified Euler--Lagrange equations, which are obtained from Noether's Theorem, become that of the original system through the same substitution $v=u$. This approach of constructing conservation laws works equally well for other variational symmetries of $\widehat{L}$, which are not necessary extended from known symmetries of the original system. Worked examples will be provided.

 To make the paper self-contained, we will review relevant  fundamental theories on symmetry analysis in Section \ref{sec:rev}, for instance, the linearized symmetry condition for determining symmetries of differential equations,  symmetries of variational problems and conservation laws of Euler--Lagrange equations obtained from Noether's Theorem, and a brief introduction to the formal Lagrangian method proposed by Ibragimov. Readers who are familiar with these topics and notations may move to Section \ref{sec:mfl} directly. In Section \ref{sec:mfl}, we define the modified formal Lagrangian formulation for a general system of differential equations.  Algorithms for extending symmetries of a  system of differential equations to variational symmetries of the corresponding modified formal Lagrangian are given too. The viscous Burgers' equation is used as an illustrative running example. Further concrete examples will be studied in Section \ref{sec:con}, including the derivation of a nontrivial conservation law for the Fornberg--Whitham equation using a symmetry extended from its Lie point symmetries and a correspondence between variational symmetries of modified formal Lagrangians and conservation laws of fluid equations.

\section{A review of symmetries, conservation laws and Noether's Theorem}\label{sec:rev}
In this section, we  briefly review the linearized symmetry condition for computing symmetries of differential equations and conservation laws obtained through Noether's Theorem; details can be found in, e.g., Olver's book \cite{Olver1993}. Ibragimov's formal Lagrangian approach for computing conservation laws will also be reviewed. 

\subsection{The linearized symmetry condition}
For a system of differential equations, let $x=(x^1,x^2,\ldots,x^p)\in \mathbb{R}^p$ be the independent variables and  let $u=(u^1,u^2,\ldots,u^q)\in\mathbb{R}^q$ be the dependent variables. In the examples though, we will use $t$ and $x$ to denote the  time and space  as the independent variables. 
A system of differential equations is defined on the jet bundles (cf. \cite{KraVin1999,Kup1980,Sau1989}) coordinated with
\begin{equation*}
(x,[u])
\end{equation*} 
where $[u]$ denotes $u$ and sufficiently many of their derivatives, written in terms of the multi-index notations 
 \begin{equation*}
u^{\alpha}_{\bold{J}}=\frac{\partial^{|\bold{J}|} u^{\alpha}}{\partial (x^1)^{j_1}\partial (x^2)^{j_2}\ldots \partial (x^p)^{j_p}},
\end{equation*} 
where $\bold{J}=(j_1,j_2,\ldots,j_p)$ and $|\bold{J}|=j_1+j_2+\cdots +j_p$. Each index $j_i$ is a non-negative integer, denoting the number of total derivatives with respect to the independent variable $x^i$. Therefore, a system of differential equations can be written locally as
\begin{equation}\label{eq:de0}
\mathcal{A}=\left\{F_{\alpha}(x,[u])=0,\quad \alpha=1,2,\ldots,q\right\}.
\end{equation}
Note than we assumed  that the number of equations in the system \eqref{eq:de0} is the same as the dimension of dependent variables $u$. 
For simplicity, we will often assume  that the system is analytic and totally nondegenerate, the latter of which means the system itself and its prolongations are of maximal rank and locally solvable; see, e.g., Definition 2.83 in \cite{Olver1993}. 

Consider the following invertible local transformations around $\varepsilon=0$:
\begin{equation}\label{eq:tran}
\begin{aligned}
&x\mapsto \widetilde{x}(x,u,\varepsilon),\quad 
u\mapsto \widetilde{u}(x,u,\varepsilon), \\
\text{ subject to }  ~&\widetilde{x}(x,u,0)=x,\quad \widetilde{u}(x,u,0)=u,
\end{aligned}
\end{equation} 
 which can be prolonged to the derivatives $u^{\alpha}_{\bold{J}}$. In practice, it is often more convenient to prolong the corresponding infinitesimal generator as introduced below. The transformations form a local symmetry group of the system \eqref{eq:de0} if and only if the group maps one solution $u=f(x)$ to another solution $\widetilde{u}=\widetilde{f}(\widetilde{x})$. It is often more convenient to use the corresponding {\bf infinitesimal generator }
\begin{equation}\label{eq:ig}
X=\xi^i(x,u)\frac{\partial}{\partial {x^i}}+\phi^{\alpha}(x,u)\frac{\partial}{\partial {u^{\alpha}}},
\end{equation}
where 
\begin{equation*}
\xi^i(x,u):=\frac{\operatorname{d}}{\operatorname{d}\!\varepsilon}\Big|_{\varepsilon=0}\widetilde{x}^i,\quad \phi^{\alpha}(x,u):=\frac{\operatorname{d}}{\operatorname{d}\!\varepsilon}\Big|_{\varepsilon=0}\widetilde{u}^{\alpha}.
\end{equation*}
Note that the Einstein summation convention is used from now on.
Prolongation of the transformations \eqref{eq:tran} to higher jets yields  prolongation of the infinitesimal generator given by (cf. Theorem 2.36 in \cite{Olver1993})
\begin{equation}
\operatorname{pr}\!X=\xi^iD_i+Q^{\alpha}\frac{\partial}{\partial u^{\alpha}}+\cdots+(D_{\bold{J}}Q^{\alpha})\frac{\partial}{\partial u^{\alpha}_{\bold{J}}}+\cdots.
\end{equation}
The tuple $Q$ with $Q^{\alpha}=\phi^{\alpha}-\xi^jD_iu^{\alpha}$ is called the {\bf characteristic} of $X$ and $D_i$ is the total derivative with respect to $x^i$:
\begin{equation*}
D_i:=\frac{\partial}{\partial x^i}+u^{\alpha}_{\bold{1}_i}\frac{\partial }{\partial u^{\alpha}}+\cdots+u^{\alpha}_{\bold{J}+\bold{1}_i}\frac{\partial }{\partial u^{\alpha}_{\bold{J}}}+\cdots,
\end{equation*}
where $\bold{1}_i$ is the $p$-tuple with only one nonzero entry $1$ at the $i$-th place. The multi-index notation $D_{\bold{J}}$ denotes a multiple number of total derivatives:
\begin{equation*}
D_{\bold{J}}=D_1^{j_1}D_2^{j_2}\cdots D_p^{j_p}, \quad \bold{J}=(j_1,j_2,\ldots,j_p).
\end{equation*}

For a system of differential equations \eqref{eq:de0} satisfying the nondegeneracy condition, i.e., of maximal rank and locally solvable,  a vector field $X$ generates a symmetry group of local transformations if and only if the {\bf linearized symmetry condition} is satisfied (e.g., \cite{Hydon2000,Olver1993}), namely 
\begin{equation}
\operatorname{pr}\!X(F_{\alpha})=0,\quad \alpha=1,2,\ldots,q, \text{ whenever the system \eqref{eq:de0} holds}.
\end{equation}
The total nondegeneracy  and analyticity conditions further allow us to express  the linearized symmetry condition 
equivalently to the existence of $q\times q $ matrices $(K^{\bold{J}}(x,[u]))$ whose entries are functions of $(x,[u])$ such that 
\begin{equation}\label{eq:lsc}
\operatorname{pr}\!X(F_{\alpha})=\sum_{\beta,\bold{J}}K^{\bold{J}}_{\alpha\beta}\left(D_{\bold{J}}F_{\beta}\right), \quad \alpha=1,2,\ldots,q.
\end{equation}

\begin{rem}
Symmetries corresponding to an infinitesimal generator of the form \eqref{eq:ig} are called Lie point symmetries. They will be called generalised symmetries when coefficients of the infinitesimal generator not only depend on $x$ and $u$ but also on derivatives of $u$, namely when the infinitesimal generator is of the form
\begin{equation*}
X=\xi^i(x,[u])\frac{\partial}{\partial x^i}+\phi^{\alpha}(x,[u])\frac{\partial}{\partial u^{\alpha}}.
\end{equation*}
\end{rem}

\begin{exm} \label{exm:kdv1}
The Korteweg--de Vries (KdV) equation
\begin{equation}\label{eq:kdv}
u_t+uu_x+u_{xxx}=0
\end{equation}
admits  a four-dimensional group of Lie point symmetries generated by 
\begin{equation*}
X_1=\partial_t,\quad X_2=\partial_x,\quad X_3=t\partial_x+\partial_u,\quad X_4=3t\partial_t+x\partial_x-2u\partial_u.
\end{equation*}
\end{exm}

\subsection{Variational symmetries and Noether's Theorem}

Consider a variational problem with a functional 
\begin{equation}\label{eq:vp}
\mathscr{L}[u]=\int_{\Omega} L(x,[u])\operatorname{d}\!x
\end{equation}
defined in an open, connected subspace $\Omega$ with smooth boundary,
where  the smooth function $L(x,[u])$ is called a Lagrangian (density function). 
Variational calculus leads to the Euler--Lagrange equations $\bold{E}_{u^\alpha}(L)=0$, $\alpha=1,2,\ldots,q$, which are written using the {\bf Euler operators}  (e.g., Section 4.1 in \cite{Olver1993})
\begin{equation}\label{eq:eo}
\begin{aligned}
\bold{E}_{u^{\alpha}}:&=\sum_{\bold{J}}(-D)_{\bold{J}}\frac{\partial}{\partial u^{\alpha}_{\bold{J}}}\\
&=\frac{\partial}{\partial u^{\alpha}}-D_i\frac{\partial}{\partial u^{\alpha}_{\bold{1}_i}}+D_iD_j\frac{\partial}{\partial u^{\alpha}_{\bold{1}_i+\bold{1}_j}}-\cdots,
\end{aligned}
\end{equation}
where $(-D)_{\bold{J}}$ is the adjoint of the operator $D_{\bold{J}}$: $(-D)_{\bold{J}}=(-1)^{|\bold{J}|}D_{\bold{J}}$.

Invariance of the variational problem \eqref{eq:vp}
  with respect to the transformations \eqref{eq:tran} can be expressed as the  {\bf infinitesimal invariance  criterion} (cf. \cite{Olver1993}, Definition 4.33)
\begin{equation}\label{eq:iic}
\operatorname{pr}\!X(L)+LD_i\xi^i=\operatorname{Div}A
\end{equation} 
for some $p$-tuple $A(x,[u])$, where $X$ is the corresponding infinitesimal generator \eqref{eq:ig}.  In fact, this can be extended to generalised symmetries equally.
Such symmetries generated by $X$ are called (divergence) {\bf variational symmetries.}  The divergence of a $p$-tuple $A$ is defined as
\begin{equation*}
\operatorname{Div}A:=D_{i}A^i.
\end{equation*}

A {\bf conservation law} of a system \eqref{eq:de0} is a divergence expression of a $p$-tuple $P(x,[u])$
\begin{equation*}
\operatorname{Div}P=0
\end{equation*} 
that vanishes on solutions of the system. Conservation laws can be trivial in two ways: The first kind is that the $p$-tuple $P$ itself vanishes on solutions of the system, e.g., when each component $P^{i}$ is a linear combination of the equations $\{F_{\alpha}\}$; The second kind is $\operatorname{Div}P\equiv 0$ holds identically for all functions $u=f(x)$. 
 In particular, for a totally nondegenerate system of differential equations, a conservation law can be equivalently understood as the existence of functions $M_{\alpha}^{\bold{J}}(x,[u])$ such that
\begin{equation}
\operatorname{Div}P=\sum_{\alpha,\bold{J}}M^{\bold{J}}_{\alpha}(D_{\bold{J}}F_{\alpha}),
\end{equation}
which can be integrated by parts to yield its characteristic form 
\begin{equation}
\operatorname{Div}\widehat{P}=Q^{\alpha}F_{\alpha}.
\end{equation}
Here $Q(x,[u])$ is the {\bf characteristic} of the equivalent conservation laws $P$ and $\widehat{P}$.  If further the system is analytic, then a conservation law is trivial if and only if its characteristic $Q$ is trivial; a trivial characteristic is defined as
\begin{equation}
Q=0 \text{ holds on solutions of the system \eqref{eq:de0}}. 
\end{equation}
When a characteristic is given, the divergence form can be derived  by using the homotopy operator on the variational bicomplex providing the cohomology is trivial (e.g., \cite{And1989,KraVin1999}) or by intuition. A general formula is  available in \cite{Olver1993} (see Section 5.4 therein).

For a variational problem, Noether's Theorem establishes a one-to-one correspondence between  variational symmetries and conservation laws of the Euler--Lagrange equations.

\begin{thm}[{\bf Noether's Theorem}]
Suppose that a vector field 
\begin{equation*}
X=\xi^i(x,[u])\frac{\partial} {\partial{x^i}}+\phi^{\alpha}(x,[u])\frac{\partial}{\partial {u^{\alpha}}}
\end{equation*}
satisfies the infinitesimal invariance criterion \eqref{eq:iic} for a variational problem \eqref{eq:vp}. Then its characteristic $Q^{\alpha}=\phi^{\alpha}(x,[u])-\xi^i(x,[u])u^{\alpha}_i$ is also the characteristic of a conservation law for the corresponding Euler--Lagrange equations $\bold{E}_{u^{\alpha}}(L)=0$. Namely, there exists a $p$-tuple $P(x,[u])$ such that
\begin{equation}
\operatorname{Div}P=Q^{\alpha}\bold{E}_{u^{\alpha}}(L).
\end{equation}
\end{thm}

Various proofs can be found in different contexts. We briefly review Olver's proof (\cite{Olver1993}, Theorem 4.29) by integrating the identity \eqref{eq:iic} by parts:
\begin{equation}\label{eq:ibp}
\begin{aligned}
\operatorname{Div}A&=\operatorname{pr}\!X(L)+LD_i\xi^i\\
&=\xi^iD_iL+\sum_{\alpha,\bold{J}}(D_{\bold{J}}Q^{\alpha})\frac{\partial L}{\partial u^{\alpha}_{\bold{J}}}+LD_i\xi^i\\
&=D_i(L\xi^i)+Q^{\alpha}\sum_{\bold{J}}(-D)_{\bold{J}}\frac{\partial L}{\partial u^{\alpha}_{\bold{J}}}+\operatorname{Div}B
\end{aligned}
\end{equation} 
for some $p$-tuple $B(x,[u])$. The resulting conservation law is
\begin{equation}
\operatorname{Div}P=Q^{\alpha}\bold{E}_{u^{\alpha}}(L), \text{ where }  P=A-L\xi-B.
\end{equation}

\begin{exm}
Let us consider the  $1+1$-dimensional  linear wave equation $u_{tt}-c^2u_{xx}=0$ as an illustrative example, where $c\neq 0$. The corresponding Lagrangian is 
\begin{equation*}
L(u_t,u_x)=-\frac{1}{2}u_t^2+\frac{c^2}{2}u_x^2
\end{equation*}
and the wave equation can be equivalently written as $\bold{E}_u(L)=0$. Since the Lagrangian
is explicitly independently from $t$ and $x$, it admits the time translational and space translational symmetries, namely 
\begin{equation*}
t\mapsto t+\varepsilon_1,\quad x\mapsto x+\varepsilon_2,
\end{equation*}
whose infinitesimal generators are respectively
\begin{equation*}
\partial_t,\quad \partial_x.
\end{equation*}
The corresponding conservation laws are written in terms of the characteristics $-u_t$ and $-u_x$, namely
\begin{equation*}
\begin{aligned}
D_t\left(-\frac{1}{2}u_t^2-\frac{c^2}{2}u_x^2\right)+D_x\left(c^2u_tu_x\right)&=-u_t(u_{tt}-c^2u_{xx}),\\
D_t\left(-u_tu_x\right)+D_x\left(\frac{1}{2}u_t^2+\frac{c^2}{2}u_x^2\right)&=-u_x(u_{tt}-c^2u_{xx}).
\end{aligned}
\end{equation*}
Note that in the integration by parts formula \eqref{eq:ibp}, $A\equiv 0$ for both infinitesimal generators. 


\end{exm}

\subsection{Formal Lagrangians and self-adjointness}
In this subsection, we brief review formal Lagrangians and the self-adjointness approach for computing conservation laws.

By introducing dummy dependent variables $v$ with the same dimension of the dependent variables $u$, the {\bf formal Lagrangian} for a system of differential equations \eqref{eq:de0} is defined as 
\begin{equation}
L(x,[u;v]):=v^{\alpha}F_{\alpha}(x,[u]).
\end{equation}
The corresponding Euler--Lagrange equations consist of two parts, namely the original system 
\begin{equation}\label{eq:de00}
0=\bold{E}_{v^{\alpha}}(L)\equiv F_{\alpha}(x,[u]), \quad \alpha=1,2,\ldots, q,
\end{equation}
and the so-called adjoint system 
\begin{equation}\label{eq:ade0}
\begin{aligned}
0=\bold{E}_{u^{\alpha}}(L):=F_{\alpha}^{\ast}(x,[u;v]), \quad \alpha=1,2,\ldots, q.
\end{aligned}
\end{equation}
The system \eqref{eq:de00} is said to be (quasi) {\bf self-adjoint} if the adjoint system \eqref{eq:ade0} is equivalent to itself via a proper substitution $v=h(u)$. In other words, Euler--Lagrange equations governed by the formal Lagrangian reduce to the original system via the substitution $v=h(u)$.

\begin{rem}
It was realised that many systems are not self-adjoint through a substitution $v=h(u)$. In recent years, there have been generalisations  to, for instance, weak self-adjointness  and nonlinear self-adjointness (e.g., \cite{Ibr2011a,Gan2011}), that, however, have been found restricted in deriving nontrivial conservation laws (e.g., \cite{HHV2014}); such an example is studied in Section \ref{subsec41}. Here, we only introduce the simplest case of self-adjointness.
\end{rem}


An important observation by Ibragimov \cite{Ibr2007} is that any symmetry generator $X$ of the original system \eqref{eq:de00} can be extended to a variational symmetry generator $X+\phi_*^{\alpha}\partial_{v^{\alpha}}$ for the formal Lagrangian, yielding a conservation law of the Euler--Lagrange equations \eqref{eq:de00} and \eqref{eq:ade0} using Noether's Theorem. Taking the self-adjointness condition into consideration, this conservation law becomes a conservation law of the original system \eqref{eq:de00} via the substitution $v=h(u)$. Although non-triviality and completeness of so-obtained conservation laws are not promised \cite{Anco2017}, its simplicity for implementation is a great advantage while on the other side the formal Lagrangian structure provides necessary foundations  for constructing variational integrator \cite{KM2015}. We will illustrate the algorithm by considering the KdV equation \eqref{eq:kdv} as an example.


\begin{exm} {\emph {(KdV equation continued.)}} The formal Lagrangian for the KdV equation is
\begin{equation*}
L=vF \text{ where }  F=u_t+uu_x+u_{xxx},
\end{equation*}
and the adjoint equation is
\begin{equation*}
F^*:=-v_t-uv_x-v_{xxx}=0.
\end{equation*}
The adjoint equation turns into the KdV equation through the substitution $v=u$:
\begin{equation*}
F^*\Big|_{v=u}=-F.
\end{equation*}

Lie point symmetries of the KdV equation were given in Example \ref{exm:kdv1}.
Here we compute the  conservation law corresponding to the scaling symmetry $$X_4=3t\partial_t+x\partial_x-2u\partial_u.$$ The extended variational symmetry for the formal Lagrangian is generated by 
\begin{equation*}
Y_4=X_4+v\frac{\partial}{\partial v}.
\end{equation*}
By using the characteristics $Q^u=-2u-3tu_t-xu_x$ and $Q^v=v-3tv_t-xv_x$, the 
 conservation law is written in characteristic form
$$D_tP^t(t,x,[u;v])+D_xP^x(t,x,[u;v])=Q^uF^*+Q^vF.$$ 
Substituting $v=u$ inside, we obtain a conservation law of the KdV equation, the conservation of momentum, as follows
\begin{equation*}
D_t\left(\frac{1}{2}u^2\right)+D_x\left(\frac13u^3+uu_{xx}-\frac{1}{2}u_x^2\right)=uF.
\end{equation*}
\end{exm}

\section{The modified formal Lagrangian formulation}\label{sec:mfl}
The formal Lagrangian method  defined by Ibragimov is  limited even when we are restricted to evolutionary equations. Nonlinear improvements have been introduced but they are, in many situations, case-by-case, in particular to determine the  substitution of dummy variables.   In this section, we  propose a modification of formal Lagrangians that is applicable to any differential equations, that we will call the modified formal Lagrangian formulation. 
Let us start with a motivating example, the viscous Burgers' equation, which serves as a running example in this section.

\begin{exm}
\label{exm:be0}
The viscous Burgers' equation reads 
\begin{equation*}
u_t+uu_x-au_{xx}=0;
\end{equation*}
we assume that the viscosity $a$ is nonzero. 
By using the usual formal Lagrangian method, the adjoint equation for a dummy variable $v$ can be  calculated from the formal Lagrangian $v\left(u_t+uu_x-au_{xx}\right)$, namely
\begin{equation*}
-v_t-uv_x-av_{xx}=0,
\end{equation*}
which is not equivalent to the viscous Burgers' equation via any substitution $v=h(u)$.

However, if we modify the formal Lagrangian by adding an extra term $-au_x^2$ and define a modified formal Lagrangian as follows
\begin{equation*}
\widehat{L}:=v\left(u_t+uu_x-au_{xx}\right)+auu_{xx},
\end{equation*}
the corresponding  (modified) adjoint equation is equivalent to the viscous Burgers' equation via the substitution $v=u$. In fact, the modified Euler--Lagrange equations consist of two parts:  variation w.r.t. $v$ gives the viscous Burgers' equation, while variation w.r.t. $u$ gives the  adjoint equation, reading
\begin{equation*}
0=\bold{E}_u(\widehat{L})\equiv -v_t-uv_x-av_{xx}+2au_{xx}.
\end{equation*}
Substituting $v=u$ inside  gives an equation differing with the Burgers' equation by a minus sign.
\end{exm}

This example motivates the definition of a modified formal Lagrangian formulation below.

\begin{defn}
 For a system of differential equations \eqref{eq:de0}, namely, 
 $$\mathcal{A}=\left\{F_{\alpha}(x,[u])=0,\quad \alpha=1,2,\ldots,q\right\},$$
 introduce dummy dependent variables $v\in\mathbb{R}^q$. If there exists a function $L_0(x,[u])$ such that the Euler--Lagrange equations governed by the Lagrangian 
\begin{equation}\label{eq:mflag}
\widehat{L}(x,[u;v])=v^{\alpha}F_{\alpha}(x,[u])+L_0(x,[u])
\end{equation}
reduce to the original system \eqref{eq:de0} via the substitution $v=u$, then we call the Lagrangian $\widehat{L}(x,[u;v])$ a {\bf modified formal Lagrangian} and the corresponding function $L_0(x,[u])$ a {\bf balance function}.
\end{defn}

To distinguish from formal Lagrangians, we will use $\widehat{L}$ to denote a modified formal Lagrangian in the current paper.
 There are several fundamentally important remarks or facts regarding the modification. Some of them are as follows.
\begin{itemize}
\item A first remark is that the substitution can be chosen arbitrary as $v=h(x,[u])$ where $h(x,[u])$ are arbitrary functions, but $v=u$ is among the simplest ones such that the adjoint system is equivalent to the original system.
\item Secondly, since the balance function $L_0$ is independent from $v$, half of the modified Euler--Lagrange equations, i.e., $\bold{E}_{v^{\alpha}}(\widehat{L})=0$, is exactly the original system.
\item  A balance function exists for any system of differential equations but not necessarily uniquely. The existence will be proved in Theorem \ref{thm:mflag}. It is not unique due to the existence of null Lagrangians, namely functions written in a divergence form; see, e.g.,  \cite{Olver1993,Olver1995}.
\item When the balance function can be written in a divergence form, the corresponding modified formal Lagrangian becomes a formal Lagrangian, namely without modification. 
\end{itemize}

\begin{thm}
\label{thm:mflag} 
For any system of differential equations
 $$\mathcal{A}=\left\{F_{\alpha}(x,[u])=0,\quad \alpha=1,2,\ldots,q\right\},$$
there exists a generic modified formal Lagrangian 
 \begin{equation}\label{eq:mflde}
 \begin{aligned}
\widehat{L}(x,[u;v]):&=v^{\alpha}F_{\alpha}(x,[u])-u^{\alpha}F_{\alpha}(x,[u])\\
&=\left(v^{\alpha}-u^{\alpha}\right)F_{\alpha}(x,[u]).
\end{aligned}
 \end{equation}
The function $L_0(x,[u])=-u^{\alpha}F_{\alpha}(x,[u])$ will be called a generic balance function.
\end{thm}

\begin{proof}
We only need to show that the corresponding  Euler--Lagrange equations reduce to the original system $\mathcal{A}$ via the substitution $v=u$. Now the modified Euler--Lagrange equations read
\begin{equation}
\begin{aligned}
0&=\bold{E}_{v^{\alpha}}(\widehat{L})\equiv F_{\alpha}(x,[u]),\\
0&=\bold{E}_{u^{\alpha}}(\widehat{L}):=\widehat{F}^*_{\alpha}(x,[u;v]).
\end{aligned}
\end{equation}
Direct computation expands the modified adjoint system as follows
\begin{equation}
\begin{aligned}
\widehat{F}^*_{\alpha}&=\bold{E}_{u^{\alpha}}\!\left(v^{\beta}F_{\beta}-u^{\beta}F_{\beta}\right)\\
&=\sum_{\beta,\bold{J}}(-D)_{\bold{J}}\left(\left(v^{\beta}-u^{\beta}\right)\frac{\partial F_{\beta}}{\partial u^{\alpha}_{\bold{J}}}\right)-F_{\alpha},
\end{aligned}
\end{equation}
that obviously reduces to the original system with the substitution $v=u$, i.e.,
\begin{equation}\label{eq:FF}
\widehat{F}^*_{\alpha}\Big|_{v=u}=-F_{\alpha}.
\end{equation}
This completes the proof.
\end{proof}

\begin{rem}
The relation between modified adjoint system and the adjoint system \eqref{eq:ade0} (without modification)  is 
\begin{equation}
\widehat{F}^*_{\alpha}(x,[u;v])=F^*_{\alpha}(x,[u;v])+\bold{E}_{u^{\alpha}}(L_0).
\end{equation}
\end{rem}

In practice, the balance function may include total derivative terms that we often prefer to mod out since they have no contribution in the Euler--Lagrange equations. For instance, the balance function $auu_{xx}$ for the viscous Burgers' equation in Example \ref{exm:be0} is equivalent to the generic one $-u(u_t+uu_x-au_{xx})$ by differing a divergence
\begin{equation}
u(u_t+uu_x-au_{xx})+auu_{xx}=D_t\left(\frac12u^2\right)+D_x\left(\frac13u^3\right).
\end{equation}
Another convenient and equivalent choice is $-au_{x}^2$; see Equation \eqref{eq:bbbb}.

Since there exists a modified formal Lagrangian for any system of differential equations, it would be interesting to consider some well-known examples. 
\begin{itemize}
\item Evolutionary equations
\begin{equation}
u_t^{\alpha}=f^{\alpha}(x,t,[u]_x),\quad \alpha=1,2,\ldots,q,
\end{equation}
where the short hand notation $[u]_x$ denotes $u$ and finitely many of their derivatives w.r.t. to $x$ only. Note that $x$ can be multi-dimensional. The generic modified formal Lagrangian reads
\begin{equation}
\widehat{L}=\sum_{\alpha}v^{\alpha}\left(u_t^{\alpha}-f^{\alpha}\right)-\sum_{\alpha}u^{\alpha}\left(u^{\alpha}_t-f^{\alpha}\right).
\end{equation}
An equivalent modified formal Lagrangian is
\begin{equation}
\widehat{L}=\sum_{\alpha}v^{\alpha}\left(u_t^{\alpha}-f^{\alpha}\right)+\sum_{\alpha}u^{\alpha}f^{\alpha}.
\end{equation}
Note that a special case was considered in \cite{KM2015} (Equations (114-115) therein).
\item A family of Camassa--Holm-type equations
\begin{equation}
u_t-\epsilon u_{xxt}=g(x,t,[u]_x),\quad \varepsilon \neq 0.
\end{equation}
The generic modified formal Lagrangian reads
\begin{equation}
\widehat{L}=v\left(u_t-\epsilon u_{xxt}-g\right)-u\left(u_t-\epsilon u_{xxt}-g\right),
\end{equation}
which is equivalent to 
\begin{equation}
\widehat{L}=v\left(u_t-\epsilon u_{xxt}-g\right)+\epsilon uu_{xxt}+ug.
\end{equation}
\end{itemize}
For concrete examples, further divergence terms can appear and they can also be modded out.

Next, we are going to show the connections between symmetries of the original system and variational symmetries of the modified formal Lagrangian. Such connections allow us to derive  conservation laws of the modified Euler--Lagrange equations using Noether's Theorem, that can amount to conservation laws of the original system.

\begin{thm}\label{thm:sym1}
Consider a system of differential equations
 $$\mathcal{A}=\left\{F_{\alpha}(x,[u])=0,\quad \alpha=1,2,\ldots,q\right\},$$
 that is totally nondegenerate and analytic, and that admits a symmetry generated by 
 \begin{equation*}
X=\xi^i(x,[u])\frac{\partial} {\partial{x^i}}+\phi^{\alpha}(x,[u])\frac{\partial}{\partial {u^{\alpha}}}.
 \end{equation*} 
 Then $X$ can be extended to a variational symmetry
 $$Y=X+\phi_*^{\alpha}(x,[u;v])\frac{\partial}{\partial v^{\alpha}}$$
  of the generic modified formal functional $$\widehat{\mathscr{L}}[u;v]:=\int_{\Omega}\widehat{L}(x,[u;v])\operatorname{d}\!x,$$ where the functions $\phi_*$ are to be determined and the generic modified formal Lagrangian is
 \begin{equation*}
\widehat{L}(x,[u;v]):=v^{\alpha}F_{\alpha}(x,[u])-u^{\alpha}F_{\alpha}(x,[u]).
 \end{equation*}
\end{thm}

\begin{proof}
First of all, as the system is totally nondegenerate and analytic, the linearized symmetry condition is replaced by \eqref{eq:lsc}, namely
\begin{equation*}
\operatorname{pr}\!X(F_{\alpha})=\sum_{\beta,\bold{J}}K^{\bold{J}}_{\alpha\beta}\left(D_{\bold{J}}F_{\beta}\right), \quad \alpha=1,2,\ldots,q,
\end{equation*}
for some functions $K_{\alpha\beta}^{\bold{J}}(x,[u])$. The extended infinitesimal generator $Y$  satisfies the infinitesimal invariance criterion for the modified formal functional, that is,
\begin{equation*}
\begin{aligned}
\operatorname{pr}\!Y(\widehat{L})+\widehat{L}D_i\xi^i=\operatorname{Div}A
\end{aligned}
\end{equation*}
for some $p$-tuple $A(x,[u;v])$. Its left-hand side can be integrated by parts as follows 
\begin{equation}\label{eq:ybb}
\begin{aligned}
\operatorname{pr}\!Y&(\widehat{L})+\widehat{L}D_i\xi^i=\operatorname{pr}\!X(\widehat{L})+\phi_*^{\alpha}F_{\alpha}+\widehat{L}D_i\xi^i\\
&=\left(v^{\alpha}-u^{\alpha}\right)\operatorname{pr}\!X(F_{\alpha})-\phi^{\alpha}F_{\alpha}+\phi_*^{\alpha}F_{\alpha}+\left(v^{\alpha}-u^{\alpha}\right)(D_i\xi^i)F_{\alpha}\\
&=\sum_{\alpha,\beta,\bold{J}}\left(v^{\alpha}-u^{\alpha}\right)K^{\bold{J}}_{\alpha\beta}\left(D_{\bold{J}}F_{\beta}\right)-\phi^{\alpha}F_{\alpha}+\phi_*^{\alpha}F_{\alpha}+\left(v^{\alpha}-u^{\alpha}\right)(D_i\xi^i)F_{\alpha}\\
&=\sum_{\alpha}\Big\{(-D)_{\bold{J}}\left[\left(v^{\beta}-u^{\beta}\right)K^{\bold{J}}_{\beta\alpha}\right]-\phi^{\alpha}+\phi_*^{\alpha}+\left(v^{\alpha}-u^{\alpha}\right)(D_i\xi^i)\Big\}F_{\alpha}+\operatorname{Div}B
\end{aligned}
\end{equation}
for some $p$-tuple $B(x,[u;v])$. Clearly, the undetermined functions $\phi_*$ can be chosen as 
\begin{equation}\label{eq:pp*}
\phi_*^{\alpha}=\phi^{\alpha}-(-D)_{\bold{J}}\left[\left(v^{\beta}-u^{\beta}\right)K^{\bold{J}}_{\beta\alpha}\right]-\left(v^{\alpha}-u^{\alpha}\right)(D_i\xi^i)
\end{equation}
and consequently $A=B$. This finishes the proof.

\end{proof}

Theorem \ref{thm:sym1} implies that any symmetry of the original system amounts to a 
 conservation law of the Euler--Lagrange equations governed by the modified formal Lagrangian.  
However, be noted that the extension of symmetries may  not be unique and the choice in Theorem \ref{thm:sym1}, i.e., Equation \eqref{eq:pp*}, is in fact not the ideal one, because the  conservation law corresponding to the so-extended generator $Y$ becomes a trivial conservation law of the original system  when the substitution $v=u$ is applied:
\begin{equation}\label{eq:pptr}
\phi_*^{\alpha}\Big|_{v=u}=\phi^{\alpha} \text{ and hence } Q^{v^{\alpha}}\Big|_{v=u}=Q^{u^{\alpha}},
\end{equation}
and then we have
\begin{equation}
\begin{aligned}
\operatorname{Div}P\Big|_{v=u}&=\left(Q^{u^{\alpha}}\widehat{F}^*_{\alpha}+Q^{v^{\alpha}}F_{\alpha}\right)\Big|_{v=u}\\
&=\left(-Q^{u^{\alpha}}+Q^{v^{\alpha}}\Big|_{v=u}\right)F_{\alpha}\\
&=0.
\end{aligned}
\end{equation}
The relation \eqref{eq:FF} is applied here.

Fortunately, the extension to a variational symmetry $Y$ may not be unique, particularly when $-\phi^{\alpha}F_{\alpha}$ can be written in divergence form and hence can be moved into the divergence  $\operatorname{Div}B$ in \eqref{eq:ybb}. In fact, to derive a nontrivial conservation law for the original system, we must choose those extensions such that \eqref{eq:pptr} can not happen. Let us consider the running example again.

\begin{exm}{\emph {(The viscous Burgers' equation continued.)}} \label{exm:be1} 
Symmetries of the viscous Burgers' equation (see Example \ref{exm:be0}) can be calculated using the linearized symmetry condition \eqref{eq:lsc} and  its Lie point symmetries  are generated by the following infinitesimal generators
\begin{equation*}
\begin{aligned}
&X_1=\partial_t,\quad X_2=\partial_x,\quad X_3=t\partial_x+\partial_u,\\
&X_4=2t\partial_t+x\partial_x-u\partial_u,\quad X_5=t^2\partial_t+tx\partial_x+(x-tu)\partial_u.
\end{aligned}
\end{equation*}
The generic modified formal Lagrangian reads
\begin{equation*}
\widehat{L}=vF-uF, \text{ where } F:=u_t+uu_x-au_{xx}.
\end{equation*}
The modified adjoint equation is $\widehat{F}^*=0$ where
\begin{equation*}
\widehat{F}^*=-v_t-uv_x-av_{xx}+2au_{xx}.
\end{equation*}

Take $X_3$ as an example. Direct calculation shows that
\begin{equation*}
\operatorname{pr}\!X_3(F)\equiv 0.
\end{equation*}
Consequently, Equation \eqref{eq:ybb}  becomes 
\begin{equation}\label{eq:bybb}
\begin{aligned}
\operatorname{pr}\!Y_3(\widehat{L})&=\left(\phi_*-\phi\right)F\\
&=\left(\phi_*-1\right)F.
\end{aligned}
\end{equation}

\begin{itemize}
\item The extension \eqref{eq:pp*} gives 
\begin{equation*}
\phi_*=1, \text{ and hence } Y_3=t\partial_x+\partial_u+\partial_v.
\end{equation*}
This leads to a trivial conservation law of the viscous Burgers' equation.
\item Equation \eqref{eq:bybb} can be rearranged as follows
\begin{equation*}
\begin{aligned}
\operatorname{pr}\!Y_3(\widehat{L})&=\left(\phi_*-1\right)F\\
&=\phi_*F-D_tu-D_x\left(\frac{1}{2}u^2-au_x\right).
\end{aligned}
\end{equation*}
Therefore, we may choose $\phi_*$ as $0$ or $-1$ instead of $1$. In both cases, we obtain a nontrivial conservation law of the viscous Burgers' equation:
\begin{equation*}
D_tu+D_x\left(\frac{1}{2}u^2-au_x\right)=F.
\end{equation*}
\end{itemize}
\end{exm}

As we notice in the example above that the observation needed for obtaining nontrivial conservation laws is relatively strong. Moreover, we often prefer to Lagrangians including no null information, namely without terms written in divergence form.  The following theorem provides another approach for extending symmetries of a system of differential equations to variational symmetries of its (not necessary generic) modified formal Lagrangian; in fact, this method is often more convenient and practical, compared with Theorem \ref{thm:sym1},  for deriving nontrivial conservation laws.

\begin{thm}\label{thm:weak}
Consider a system of differential equations
 $$\mathcal{A}=\left\{F_{\alpha}(x,[u])=0,\quad \alpha=1,2,\ldots,q\right\},$$
 that is totally nondegenerate and analytic, and that admits a symmetry generated by 
 \begin{equation*}
X=\xi^i(x,[u])\frac{\partial} {\partial{x^i}}+\phi^{\alpha}(x,[u])\frac{\partial}{\partial {u^{\alpha}}}.
 \end{equation*} 
 Assume  \begin{equation*}
\widehat{L}(x,[u;v]):=v^{\alpha}F_{\alpha}(x,[u])+L_0(x,[u]).
 \end{equation*}
 is a modified formal Lagrangian of the system, such that the modified adjoint system is equivalent to the original system via the substitution $v=u$. 
 
If the balance variational problem, whose Lagrangian is the balance function $L_0(x,[u])$, is invariant w.r.t. $X$,
 then $X$ can be extended to a variational symmetry
 $$Y=X+\phi_*^{\alpha}(x,[u;v])\frac{\partial}{\partial v^{\alpha}}$$
  of the modified formal functional $$\widehat{\mathscr{L}}[u;v]:=\int_{\Omega}\widehat{L}(x,[u;v])\operatorname{d}\!x,$$ where the functions $\phi_*$ are to be determined.

\end{thm}

\begin{proof}
Since the balance variational problem 
$$\int_{\Omega}L_0(x,[u])\operatorname{d}\!x$$
 is invariant w.r.t. $X$, there exists a $p$-tuple $P_0$ such that 
\begin{equation*}
\operatorname{pr}\!X(L_0)+L_0D_i\xi^i=\operatorname{Div}P_0.
\end{equation*}
Then, we have
\begin{equation*}
\begin{aligned}
\operatorname{pr}\!Y(\widehat{L})+\widehat{L}D_i\xi^i&=\operatorname{pr}\!X\left(v^{\alpha}F_{\alpha}+L_0\right)+\phi_*^{\alpha}F_{\alpha}+\left(v^{\alpha}F_{\alpha}+L_0\right)D_i\xi^i\\
&=v^{\alpha}\operatorname{pr}\!X(F_{\alpha})+\operatorname{pr}\!X(L_0)+\phi_*^{\alpha}F_{\alpha}+v^{\alpha}(D_i\xi^i)F_{\alpha}+L_0D_i\xi^i\\
&=\sum_{\alpha,\beta,\bold{J}}v^{\alpha}K^{\bold{J}}_{\alpha\beta}\left(D_{\bold{J}}F_{\beta}\right)+v^{\alpha}(D_i\xi^i)F_{\alpha}+\phi_*^{\alpha}F_{\alpha}+\operatorname{Div}P_0\\
&=\sum_{\alpha}\Big\{(-D)_{\bold{J}}\left[v^{\beta}K_{\beta\alpha}^{\bold{J}}\right]+v^{\alpha}(D_i\xi^i)+\phi_*^{\alpha}\Big\}F_{\alpha}+\operatorname{Div}\left(B+P_0\right)
\end{aligned}
\end{equation*}
where the $p$-tuple $B(x,[u;v])$ is the consequence of integration by parts.  Therefore, the undetermined functions $\phi_*$ can be chosen as
\begin{equation}\label{eq:pweak}
\phi_*^{\alpha}=-\Big\{(-D)_{\bold{J}}\left[v^{\beta}K_{\beta\alpha}^{\bold{J}}\right]+v^{\alpha}(D_i\xi^i)\Big\}
\end{equation}
such that the modified formal functional is invariant w.r.t. $Y$.
\end{proof}

\begin{exm}{\emph {(The viscous Burgers' equation continued.)}} 
All Lie point symmetries of the viscous Burgers' equation are listed in Example \ref{exm:be1}. Let us consider the  modified formal Lagrangian given in Example \ref{exm:be0}:
\begin{equation*}
\widehat{L}=v\left(u_t+uu_x-au_{xx}\right)+auu_{xx}.
\end{equation*}
Recall that it is equivalent to the generic modified formal Lagrangian, leading to the same Euler--Lagrange equations.

The balance variational problem with Lagrangian $L_0=auu_{xx}$  is (divergence) invariant w.r.t. $X_1$, $X_2$ and $X_3$ such that 
\begin{equation*}
\operatorname{pr}\!X_i(L_0)+L_0\left(D_t\xi^t_i+D_x\xi^x_i\right)=D_xA_i,\quad i=1,2,3,
\end{equation*}
where
\begin{equation*}
A_1=0,\quad A_2=0,\quad A_3=au_x.
\end{equation*}
For each of the three infinitesimal generators, we have
\begin{equation*}
\operatorname{pr}\!X_i(F)\equiv 0,\quad i=1,2,3.
\end{equation*}

From Equation \eqref{eq:pweak}, we obtain $\phi_*=0$ for all of the three infinitesimal generators and hence $Y_i=X_i$, $i=1,2,3$.
 The corresponding characteristics are 
\begin{equation*}
\begin{aligned}
&Q_1^u=-u_t,\quad Q_1^v=-v_t,\\
&Q_2^u=-u_x,\quad Q_2^v=-v_x,\\
&Q_3^u=1-tu_x,\quad Q_3^v=-tv_x.
\end{aligned}
\end{equation*} 
The variational symmetries $Y_i$ satisfy 
\begin{equation*}
\operatorname{pr}\!Y_i(\widehat{L})+\widehat{L}\left(D_t\xi^t_i+D_x\xi^x_i\right)=D_xA_i, \quad i=1,2,3.
\end{equation*}
 
 Recall that 
\begin{equation*}
\begin{aligned}
F&=u_t+uu_x-au_{xx},\\
\widehat{F}^*&=-v_t-uv_x-av_{xx}+2au_{xx}.
\end{aligned}
\end{equation*}
The corresponding conservation laws of the modified Euler--Lagrange equations are respectively given by
\begin{equation*}
D_tP^t_i+D_xP^x_i=Q^u_i\widehat{F}^*+Q^v_iF,\quad i=1,2,3, 
\end{equation*}
where
\begin{equation*}
\begin{aligned}
&P_1^t=-uu_xv+au_{xx}v+au_x^2,\quad P^x_1=uu_tv-2au_tu_x-au_{tx}v+au_tv_x,\\
&P^t_2=-uv_x,\quad P_2^x=uv_t-au_x^2+au_xv_x,\\
&P_3^t=-(1-tu_x)v,\quad P_3^x=-uv+2au_x-av_x-tu_tv-atu_x^2+atu_xv_x.
\end{aligned}
\end{equation*}
In fact, they can be simply derived from the integration by parts formula \eqref{eq:ibp} and in this special case, we obtain them as
\begin{equation*}
P_i^t=-\widehat{L}\xi_i^t-Q_i^u\frac{\partial \widehat{L}}{\partial u_t},\quad  P_i^x=A_i-\widehat{L}\xi_i^x-Q^u_i\left(\frac{\partial \widehat{L}}{\partial u_x}-D_x\frac{\partial \widehat{L}}{\partial u_{xx}}\right)-\left(D_xQ^u_i\right)\frac{\partial \widehat{L}}{\partial u_{xx}},
\end{equation*}
where $i=1,2,3$. A general formula can be found in Section 5.4 of \cite{Olver1993}.

Setting $v=u$ in the three conservation laws, only the third one contributes to a nontrivial conservation law of the viscous Burgers' equation, namely
\begin{equation*}
D_tu+D_x\left(\frac{1}{2}u^2-au_x\right)=F,
\end{equation*}
which is the same as we obtained in Example \ref{exm:be1}. According to the analysis of cohomology of $\mathscr{C}$-spectral sequence, this is the only (local) conservation law for the viscous Burgers' equation; see, e.g., \cite{KraVin1999}. 


\end{exm}

Note that an equivalent modified formal Lagrangian for the viscous Burgers' equation can be chosen to simplify the calculations above,  which reads
\begin{equation}\label{eq:bbbb}
\widehat{L}=v\left(u_t+uu_x-au_{xx}\right)-au_{x}^2.
\end{equation}
The new balance variational problem with Lagrangian $-au_x^2$ is invariant---rather than divergence invariant---with respect to all three symmetries.

Beside the extension of known symmetries of a system to variational symmetries of its modified formal functionals through either Theorem \ref{thm:sym1} or Theorem \ref{thm:weak}, one may 
also use their own variational symmetries (not necessary extended from known symmetries) to derive  conservation laws using Noether's Theorem; see Section \ref{subsec42} for an illustrative example from fluid mechanics.

\section{Concrete examples}
\label{sec:con}

In this section, we will study some concrete examples from physics and fluid mechanics.
In the first example, we obtain a nontrivial conservation for the Fornberg--Whitham equation that has not been successfully achieved using the previous formal Lagrangian method.
For differential equations from fluid mechanics, we show how to derive conservation laws from a modified formal Lagrangian's variational symmetries, that are not necessary extended from known symmetries of the original differential equations.


\subsection{The Fornberg--Whitham equation}\label{subsec41}
The Fornberg--Whitham (FW) equation is a nonlinear dispersive wave equation, admitting a wave of greatest height, e.g., \cite{FW1978}. Symmetry analysis of a bigger family of nonlinear partial differential equations was conducted in \cite{CMP1997}. It was shown in \cite{HHV2014} (see also \cite{Ibr2011b}) that the FW equation is neither quasi self-adjoint nor weak self-adjoint through the formal Lagrangian approach; although it is nonlinearly self-adjoint but only trivial conservation laws could  be obtained. In this subsection, we will study its modified formal Lagrangian formulation to derive conservation laws.

The FW equation can be written as $F=0$ with 
\begin{equation}
F=u_t- u_{xxt}+u_x+uu_x-3u_xu_{xx}-uu_{xxx}.
\end{equation}
It admits a three-dimensional group of Lie point symmetries whose infinitesimal generators are
\begin{equation*}
X_1=\partial_t,\quad X_2=\partial_x,\quad X_3=t\partial_x+\partial_u.
\end{equation*}

Let us study the conservation law related to $X_3$ by considering the following modified formal Lagrangian
\begin{equation*}
\widehat{L}=vF+L_0(x,t,[u]),
\end{equation*}
where $v$ is the dummy dependent variable and  the balance function is chosen as
\begin{equation*}
L_0=uu_xu_{xx}.
\end{equation*}
It is equivalent to the generic one by modding out all divergence terms. Direct computation gives 
the modified Euler--Lagrange equations 
\begin{equation*}
\begin{aligned}
0&=\bold{E}_v(\widehat{L})\equiv F,\\
0&=\bold{E}_u(\widehat{L}):=\widehat{F}^*, 
\end{aligned}
\end{equation*}
where the modified adjoint equation is
\begin{equation*}
\widehat{F}^*=-v_t+v_{xxt}-v_x-uv_x+3u_xu_{xx}+uv_{xxx}
\end{equation*}
satisfying $$\widehat{F}^*\Big|_{v=u}=-F.$$

First of all, we shall check that the balance variational problem is  invariant w.r.t. $X_3$. Namely the infinitesimal invariance criterion \eqref{eq:iic} is satisfied for $L_0$; by noting $\operatorname{Div}\xi_3\equiv 0$,  we have
\begin{equation*}
\operatorname{pr}\!X_3(L_0)=D_x\left(\frac12u_x^2\right).
\end{equation*}

It can be checked that $\operatorname{pr}\!X_3(F)\equiv 0$, and hence Equation \eqref{eq:pweak} gives the extension of $X_3$ to a variational symmetry $Y_3=X_3$ of the modified formal Lagrangian, whose characteristics is written in components as 
\begin{equation*}
Q^u=1-tu_x,\quad Q^v=-tv_x.
\end{equation*}
The corresponding conservation law for the modified Euler--Lagrange equation is written in characteristic form as
\begin{equation*}
D_tP^t+D_xP^x=Q^u\widehat{F}^*+Q^vF,
\end{equation*}
where
\begin{equation*}
\begin{aligned}
&P^t=-v+v_{xx}-tuv_x+tu_{xx}v_x,\\
&P^x=-v+tuv_t+uv_{xx}-u_xv_x+\frac32u_x^2-tu_xv_{xt}-tuu_xv_{xx}+tu_x^2v_x+tuu_{xx}v_x-tu_x^3.
\end{aligned}
\end{equation*}
Substituting $v=u$ inside, it becomes a nontrivial conservation law of the FW equation written in  characteristic form as follows
\begin{equation}
D_t\left(u-u_{xx}\right)+D_x\left(u+\frac12u^2-u_x^2-uu_{xx}\right)=F.
\end{equation}

A remark on the symmetries $X_1$ and $X_2$ is that they are also variational symmetries for the balance variational problem and hence lead to nontrivial conservation laws for the modified Euler--Lagrange equations. But no new nontrivial conservation law of the FW equation can be achieved after the substitution $v=u$ is applied.




\subsection{Fluid mechanics: The incompressible Euler equations as an illustration}\label{subsec42}
In this subsection, we will show how well-known conservation laws of fluid systems represented in the Eulerian framework can be derived using the modified Lagrangian formulation. In the Lagrangian framework, variational formulation for incompressible flow has been known for quite long time.  In this paper, we study the incompressible Euler equations as an example. The same methodology applies to study other fluid equations, e.g.,  the compressible Euler equations and the compressible and incompressible Navier--Stokes equations, and other differential equations equally. Note that other variational formulations for the incompressible Euler equations exist, for instance, the Clebsch variational principle and a multisymplectic formulation (see \cite{CHH2007} and references therein for more details).

The incompressible Euler equations are the following system of partial differential equations for the velocity $u\in\mathbb{R}^n$ ($n=2$ or $3$) and the pressure $p\in\mathbb{R}$:
\begin{equation}
\begin{aligned}
&u_t+u\cdot \nabla u+\nabla p=0,\\
&\nabla\cdot u=0.
\end{aligned}
\end{equation}
Dimension of the space variable is the same as the velocity, i.e., $n=2$ or $3$, and time is one-dimensional.
This system models the flow of inviscid, incompressible fluid with constant density. Note that in this subsection, we use $n$ to denote the dimension of variables rather than $p$ (and $q$) used above as it means a different thing in fluid mechanics. Furthermore, the dummy dependent variable corresponding to the pressure $p$ will be denoted by $q$.

In Cartesian coordinates, the system can be written in component form as follows
\begin{equation}
\begin{aligned}
&\frac{\partial u^i}{\partial t}+u^j\frac{\partial u^i}{\partial x^j}+\frac{\partial p}{\partial x^i}=0,\quad i=1,2,\ldots,n,\\
&\frac{\partial u^j}{\partial x^j}=0.
\end{aligned}
\end{equation}
In this paper, we only consider the three-dimensional case, i.e., $n=3$. Introducing dummy dependent variables $v\in\mathbb{R}^3$ and $q\in\mathbb{R}$, we define the modified formal Lagrangian by
\begin{equation}\label{eq:mflineuler}
\begin{aligned}
\widehat{L}=q\left(\frac{\partial u^j}{\partial x^j}\right)+\sum_iv^i\left(\frac{\partial u^i}{\partial t}+u^j\frac{\partial u^i}{\partial x^j}+\frac{\partial p}{\partial x^i}\right)-\sum_{i,j}u^iu^j\frac{\partial u^j}{\partial x^i}.
\end{aligned}
\end{equation}
It is equivalent to the generic one and the modified Euler--Lagrange equations, consisting of the incompressible Euler equations and the adjoint equations
\begin{equation}
\begin{aligned}
&-\left(\frac{\partial v^i}{\partial t}+u^j\frac{\partial u^i}{\partial x^j}+\frac{\partial q}{\partial x^i}+\frac{\partial}{\partial x^j}\left(v^iu^j-u^iu^j\right)\right)=0,\quad i=1,2,\ldots,n,\\
&-\frac{\partial v^j}{\partial x^j}=0,
\end{aligned}
\end{equation}  reduce to the incompressible Euler equations via the substitutions $v=u$ and $q=p$.

\subsubsection{Conservation laws related to extended variational symmetries}
 It is known that when $n=3$, the system of incompressible Euler equations admits the following Lie point symmetries (e.g., Example 2.45 in \cite{Olver1993}):
\begin{itemize}
\item Moving coordinates: 
\begin{equation*}
\begin{aligned}
f_i\partial_{x^i}+f'_i\partial_{u^i}-f''_ix^i\partial_p,\quad i=1,2,3;
\end{aligned}
\end{equation*}
\item Time translation:
\begin{equation*}
\partial_t;
\end{equation*}
\item Scaling:
\begin{equation*}
\begin{aligned}
&x^i\partial_{x^i}+t\partial_t,\\
&t\partial_t-u^i\partial_{u^i}-2p\partial_p;
\end{aligned}
\end{equation*}
\item Rotations:
\begin{equation*}
x^i\partial_{x^j}-x^j\partial_{x^i}+u^i\partial_{u^j}-u^j\partial_{u^i},\quad i,j=1,2,3 \text{ and }i<j;
\end{equation*}
\item Pressure changes:
\begin{equation*}
g\partial_p.
\end{equation*}
\end{itemize}
Here the functions $f_i$ and $g$ are arbitrary functions of $t$. 

According to Theorem \ref{thm:weak}, these symmetries can be extended to variational symmetries of the modified formal Lagrangian if they are variational symmetries of the balance variational problem, whose Lagrangian is the balance function 
\begin{equation}
L_0=-\sum_{i,j}u^iu^j\frac{\partial u^j}{\partial x^i}.
\end{equation}
Using the infinitesimal invariance criterion \eqref{eq:iic},  it is immediate to verify that the balance variational problem  is invariant only w.r.t these symmetries: spatial translations  (i.e., moving coordinates with constant functions $f_i$), time translation, rotations and pressure changes.
Their infinitesimal generators $X$ can be extended to variational symmetries 
$$Y=X+\phi_{\ast}^i\partial_{v^i}+\phi_{\ast}^q\partial_q$$ 
of the modified formal Lagrangian according to Theorem \ref{thm:weak}.
\begin{itemize}
\item Extension of spatial translations  $\partial_{x^i}$:  
\begin{equation*}
Y_i=\partial_{x^i},\quad i=1,2,3;
\end{equation*}
\item Extension of time translation $\partial_t$: 
\begin{equation*}
Y=\partial_t;
\end{equation*}
\item Extension of rotations $x^i\partial_{x^j}-x^j\partial_{x^i}+u^i\partial_{u^j}-u^j\partial_{u^i}$:
\begin{equation*}
Y_{ij}=x^i\partial_{x^j}-x^j\partial_{x^i}+u^i\partial_{u^j}-u^j\partial_{u^i}+v^i\partial_{v^j}-v^j\partial_{v^i},\quad i,j=1,2,3 \text{ and } i<j;
\end{equation*}
\item Extension of pressure changes $g(t)\partial_p$: 
\begin{equation*}
Y=g(t)\partial_p.
\end{equation*}
\end{itemize}

Noether's Theorem then yields conservation laws of the modified Euler--Lagrange equations, which, by using the substitution $v=u$ and $q=p$, turn into conservation laws of the incompressible Euler equations. We only give the final results without showing intermediate computational details. Only one nontrivial conservation law is obtained after the substitution, that is 
\begin{equation*}
D_{x^i}\left(g(t)u^i\right)=g(t)\frac{\partial u^i}{\partial x^i},
\end{equation*}
corresponding to the symmetry of pressure changes.
This is  the {\bf conservation of mass}.

\subsubsection{Conservation laws related to other variational symmetries}
Except those symmetries extended from symmetries of the incompressible Euler equations, the modified formal Lagrangian \eqref{eq:mflineuler} also admits other variational symmetries. They can also be used to compute conservation laws of the incompressible Euler equations.

The first kind of infinitesimal generators is
\begin{equation}
\widehat{Y}_i=\partial_{v^i}+u^i\partial_q,\quad i=1,2,3,
\end{equation}
which are variational symmetries for the modified formal Lagrangian \eqref{eq:mflineuler}, since the infinitesimal invariance criterion \eqref{eq:iic} is satisfied for each $\widehat{Y}_i$, that is 
\begin{equation}
\begin{aligned}
\operatorname{pr}\!\widehat{Y}_i(\widehat{L})&=\left(\frac{\partial u^i}{\partial t}+u^j\frac{\partial u^i}{\partial x^j}+\frac{\partial p}{\partial x^i}\right)+u^i\left(\frac{\partial u^j}{\partial x_j}\right)\\
&=D_tu^i+D_{x^j}\left(\delta_i^jp+u^iu^j\right).
\end{aligned}
\end{equation}
The conservation laws are already written in characteristic form and they correspond to the {\bf conservation of momentum}.

Another variational symmetry is generated by
\begin{equation}
\widehat{Y}_0=u^i\partial_{v^i}+\left(\sum_i\frac12(u^i)^2+p\right)\partial_q.
\end{equation}
The conservation law is obtained using Noether's Theorem again, namely
\begin{equation}
\begin{aligned}
\operatorname{pr}\!\widehat{Y}_0(\widehat{L})&=\sum_iu^i\left(\frac{\partial u^i}{\partial t}+u^j\frac{\partial u^i}{\partial x^j}+\frac{\partial p}{\partial x^i}\right)+\left(\frac12\sum_i(u^i)^2+p\right)\left(\frac{\partial u^j}{\partial x_j}\right)\\
&=D_t\left(\frac12\sum_i(u^i)^2\right)+D_{x^i}\left(\frac12\sum_j(u^j)^2u^i+pu^i\right)\\
&=D_t\left(\frac12|u|^2\right)+\nabla\cdot \left(\frac12|u|^2u+pu\right),
\end{aligned}
\end{equation}
which is the {\bf conservation of energy}.

\section{Conclusions and future work}
A modified formal Lagrangian formulation for studying conservations laws of differential equations was defined in this paper. It was proved that any system of differential equations admits at least one modified formal Lagrangian and its self-adjointness can be achieved via the simplest substitution for the dummy variables. Practical algorithms were introduced, that allow us to extend symmetries of the original system to symmetries of its modified  formal Lagrangian and hence to compute conservation laws directly from Noether's Theorem. The same substitution for dummy variables would yield  conservation laws of the original system. We studied the viscous Burgers' equation, the Fornberg--Whitham equation and the incompressible Euler equations as illustrations.

Moreover, since the modified formal Lagrangian formulation allows us to define formally a variational structure for any system of differential equations, methods for studying variational problems can be, at least formally, applied to study non-variational differential equations, such as, symplectic/multisymplectic structures \cite{Bri1997b,MR1999a},  structure-preserving numerical methods \cite{HLW2006,MPS1998,MW2001},  invariant calculus for variational problems \cite{Man2010,MRHP2019,Peng2013}, etc., beside Noether's two theorems. In particular, we mention that variational integrator for various non-variational differential equations have been developed using (modified) formal Lagrangians including, for instance, the advection equation and the vorticity equation  \cite{KM2015}, magnetohydrodynamics \cite{K2018,KTG2016} and nonlinear dissipative wave equations \cite{Obata2020}. 

\section*{Acknowledgements} 
This work was partially supported by JSPS KAKENHI Grant Number JP20K14365, JST-CREST Grant Number JPMJCR1914, Keio Gijuku Academic Development Funds, and Keio Gijuku Fukuzawa Memorial Fund.

\end{document}